\documentclass[twocolumn,preprintnumbers,amsmath,amssymb,showpacs,superscriptaddress,prl,
]{revtex4-1}

\usepackage{comment}
\usepackage{amsmath}
\usepackage{amssymb}
\usepackage[dvips]{graphicx,epsfig}
\usepackage{latexsym}
\usepackage{pstricks}
\usepackage{pst-plot}
\usepackage{amsmath, amsthm, amsfonts, amssymb}
\usepackage{appendix}
\usepackage{endnotes}

\newtheorem{theorem}{Theorem}[section]

\begin{document}

\title{Distinguishability notion based on Wootters statistical distance: application to discrete maps}

\author{Ignacio S. Gomez}
\email{nachosky@fisica.unlp.edu.ar}
\affiliation{IFLP, UNLP, CONICET, Facultad de Ciencias Exactas, Calle 115 y 49, 1900 La Plata, Argentina}
\author{M. Portesi}
\email{portesi@fisica.unlp.edu.ar}
\affiliation{IFLP, UNLP, CONICET, Facultad de Ciencias Exactas, Calle 115 y 49, 1900 La Plata, Argentina}
\author{P. W. Lamberti}
\email{lamberti@famaf.unc.edu.ar}
\affiliation{Facultad de Matem\'{a}tica, Astronom\'{i}a y F\'{i}sica (FaMAF), Universidad Nacional de C\'{o}rdoba, Avenida Medina Allende S/N, Ciudad Universitatia, X5000HUA, C\'{o}rdoba, Argentina}

\date{\today}

\begin{abstract}
We study the distinguishability notion given by Wootters for states represented by probability density functions. This presents the particularity that it can also be used for defining a distance in chaotic unidimensional maps. Based on that definition, we provide a metric $\overline{d}$ for an arbitrary discrete map. Moreover, from $\overline{d}$ we associate a metric space to each invariant density of a given map, which results to be the set of all distinguished points when the number of iterations of the map tends to infinity. Also, we give a characterization of the wandering set of a map in terms of the metric $\overline{d}$ which allows to identify the dissipative regions in the phase space. We illustrate the results in the case of the logistic and the circle maps numerically and theoretically, and we obtain $\overline{d}$ and the wandering set for some characteristic values of their parameters.
\end{abstract}

\maketitle
\setcounter{secnumdepth}{5}

\section{Introduction}\label{sec_intro}
The concept of distance constitutes, in mathematics but also in physics, a measure of how apart two ``objects" are.
Depending on the context, a distance may refer to a physical length or to a metric function which gives place to the theory of metric spaces \cite{kolmogorov,kelley,lages}. In physics, it is well known that the information about a system is contained in the state function, the evolution of which accounts for the features of the dynamics.
When there is a limitation or uncertainty about the knowledge of the system (classical or quantum), it is common to consider the states as represented by probability density functions. Distances between probability density functions give so called \emph{statistical measures of distinguishability} between states \cite{fisher}. Many statistical distances have been defined and for several purposes, like Kullback divergence \cite{kullback}, Wootters \cite{wooters} and Monge \cite{monge} distances in quantum mechanics, or the metric distance given by the Fisher--Rao tensor in information geometry \cite{amari}.

In the context of discrete maps, one of the most important statistical features of underlying dynamics is given by the probability distribution that is a fixed point of the corresponding Frobenius--Perron operator associated to the map \cite{lasota,beck}. This is the so called invariant density which plays an important role in ergodic theory, i.e. the study of the measures that are invariant under the invariant densities \cite{walters}. Relevant properties that lay at the foundations of statistical mechanics like ergodicity and mixing are described in terms of invariant densities by means of the corresponding levels of the ergodic hierarchy \cite{walters}, with quantum extensions \cite{olimpia,nacho1,nacho2,nacho3,nacho4} that allow to characterize aspects of quantum chaos \cite{casati}.
The relevance of discrete maps lies in the fact that they serve as simple but useful models in biology, physics, economics, etc \cite{may,strogatz,lloyd}. Specifically, they have proven to be powerful tools for testing features of chaotic and complex phenomena \cite{tirnakli,tsallis,borges,csordas}.

In the case of the logistic map \cite{may,strogatz}, from the invariant density Johal \cite{johal} proposed a statistical distance which serves to characterize the chaotic regime. In this paper we explore, analytically and numerically, this idea by redefining the distance given in Ref. \cite{johal} as a metric, and we apply this to characterize two emblematic chaotic maps: the logistic map and the circle map. We focus on some values of their parameters which are in correspondence to the characteristic regimes of the dynamics. Moreover, we define a metric space associated to values of the parameters of each map that takes into account the topology of the phase space according to dynamics.

The work is organized as follows. In Section \ref{sec_prelim} we introduce the notions used throughout the paper. Section \ref{sec_metric} is devoted to
recall the proposal of a statistical distance made by Johal \cite{johal} and, using equivalence classes, we redefine it as a metric. Next, we propose a metric space, induced by this metric, composed by all the phase space points that can be distinguished in the limit of large number of iterations of the map. From the metric, we give a characterization of the wandering set of a map that allows one to classify the dissipative regions of the dynamics in phase space. Considering some characteristic values of their parameters, in Section \ref{sec_models} we illustrate the results for the logistic map and for the circle map. In section \ref{sec_conclusions} we draw some conclusions and outline future research directions.

\section{Preliminaries}\label{sec_prelim}
We recall some notions from probability theory, discrete maps and metric spaces used throughout the paper.
\subsection{Probability density and cumulative distribution functions}
A probability density function defined over an abstract space $\Gamma$ (typically, a subset of $\mathbb{R}^m$) is any nonnegative function $p:\Gamma\rightarrow\mathbb{R}_{+}$ such that
\begin{eqnarray}\label{0-1}
\int_{\Gamma}p(x)dx=1
\end{eqnarray}
where $\Gamma$ is called the \emph{space of events}. For instance, in an experiment, $\Gamma$ represents the set of all possible outcomes of the phenomenon being observed. If $\Gamma$ is composed by a discrete, say $N$, number of results then one has a discrete probability distribution which can be represented by a column vector $(p_1, p_2,\cdots ,p_N)^t$, where $p_i$ stands for the probability that the ith result occurs, and the normalization condition \eqref{0-1} now reads $\sum_{i=1}^{N}p_i=1$.

Given the probability distribution $p$ with $\Gamma\subseteq\mathbb{R}$, the so called cumulative distribution function associated to $p(x)$ is defined as
\begin{eqnarray}\label{0-2}
C(x)=\int_{\{t\leq x\}}p(t)dt
\end{eqnarray}
$C(x)$ is the probability that the variable takes a value less than or equal to $x$. When $\Gamma$ is composed by a number $N$ of results, the definition \eqref{0-2} reads as $C_i=\sum_{j=1}^{i}p_j$.
The cumulative distribution function is frequently used in statistical analysis, since an estimation of $C(x)$ can be given directly from the empirical distribution function in terms of the experimental data.
\subsection{Discrete maps and invariant densities}

Given a set $\Gamma$ and a continuous function $f:\Gamma\rightarrow\Gamma$, it is said that the sequence $\{x_n\}_{n\in\mathbb{N}_0}\subseteq\Gamma$ such that
\begin{eqnarray}\label{1-1}
x_{n+1}=f(x_n) \ \ \ \ \forall \ n\in\mathbb{N}_0 \ \ , \ \ x_0\in\Gamma
\end{eqnarray}
defines a \emph{discrete map}. More generally, one can allow $f$ to have finite discontinuities. From the physical viewpoint, a discrete map  models a system, the dynamics of which is given by iterating Eq. \eqref{1-1} where each iteration corresponds to a time step. That is, if the system is initially in a state $x_0$, then $x_n$ represents the state after $n$ time steps.
If there exists an element $c\in \Gamma$ such that $f(c)=c$, it is said that $c$, is a \emph{fixed point} of $f$.
Fixed points are physically interpreted as stationary states of the system subject to the map.

The dynamics of a discrete map can be characterized in terms of probability density functions in the following way.
The Frobenius--Perron operator $P:\mathbb{L}^1(\Gamma)\rightarrow\mathbb{L}^1(\Gamma)$ associated to the map \eqref{1-1} is given by \cite{lasota}
\begin{eqnarray}\label{1-1bis}
\int_{A}P\phi(x)dx=\int_{f^{-1}(A)}\phi(x)dx \nonumber
\end{eqnarray}
for all $\phi\in \mathbb{L}^1(\Gamma)$ and $A\subseteq\Gamma$, where $f^{-1}(A)$ is the preimage of $A$.
Any nonnegative function $\rho\in \mathbb{L}^1(\Gamma)$, normalized to $1$ in $\Gamma$, and such that
\begin{eqnarray}\label{1-2}
P\rho(x)=\rho(x) \ \ \ \ \forall \ x\in\Gamma \nonumber
\end{eqnarray}
is called \emph{invariant density} of the map. Mathematically, invariant densities are a special case of invariant measures \cite{lasota}.
In dynamical systems theory, the invariant densities represent stationary states which allows one to study the system in the asymptotic limit of large times \cite{borges2}.
As an example, in the case of the logistic map in its chaotic regime one has $f(x)=4x(1-x)$ and $\Gamma=[0,1]$. It is found that \cite[p. 7]{lasota}
\begin{eqnarray}\label{1-2frobenius}
&Pf(x)\nonumber\\
&=\frac{1}{4\sqrt{1-x}}\left\{f\left(\frac{1}{2}-\frac{1}{2}\sqrt{1-x})+f(\frac{1}{2}+\frac{1}{2}\sqrt{1-x}\right)\right\} \nonumber
\end{eqnarray}
for all $f\in\mathbb{L}^1([0,1])$. By successively applying the operator $P$ over any initial distribution $f$ (i.e. $Pf,P^2f=P(Pf),P^3f=P(P(Pf)),\ldots$) one obtains an analytical expression for the invariant density $\rho(x)$ in the limit as \cite[p. 7]{lasota}
\begin{eqnarray}\label{1-2invariant}
\rho(x)=\lim_{n\rightarrow\infty}P^nf(x)=\frac{1}{\pi\sqrt{x(1-x)}} \ \ \ , \ \ \ \ x\in(0,1) \nonumber
\end{eqnarray}
In general, except for some particular maps, the invariant density has no analytical expression, and therefore one has to compute it numerically \cite{csordas}.

The method consists of constructing an histogram to show the frequency with which states along a sequence $\{x_n\}$
fall into given regions of phase space $\Gamma$. In order to illustrate the procedure, let us consider $\Gamma=[0,1]$. Then, we simply divide the interval $[0,1]$ into $M$ discrete nonintersecting intervals of the form
\begin{eqnarray}\label{2-9bis}
[\frac{i-1}{M},\frac{i}{M})   \ \ \ \ i=1,\ldots,M \nonumber
\end{eqnarray}
Next step is to consider an initial state $x_0$ and calculate the trajectory
\begin{eqnarray}\label{2-9bisbis}
x_0 \ , \ x_1=f(x_0) \ , \ x_2=f(f(x_0)) \ , \ \ldots \ , \ x_T=f^T(x_0) \nonumber
\end{eqnarray}
of length $T$ with $T\gg M$. The fraction $\rho_i$ of the $T$ states visiting the $ith$ interval is
\begin{eqnarray}\label{2-9invariant1}
\rho_i=\frac{\sharp\{x_{\tau}\in [\frac{i-1}{M},\frac{i}{M}) \ | \ \tau=0,1,\ldots,T\}}{T} \ , \ i=1,\ldots,M \nonumber
\end{eqnarray}
Finally, the invariant density $\rho$ is given by
\begin{eqnarray}\label{2-9invariant2}
\rho(x)=\sum_{i=1}^{M}\rho_i \ \chi_{[\frac{i-1}{M},\frac{i}{M})}(x) \nonumber
\end{eqnarray}
where $\chi_{[\frac{i-1}{M},\frac{i}{M})}(x)$ denotes the characteristic function of the interval $[\frac{i-1}{M},\frac{i}{M})$.
It should be noted that in the limit $M,T\rightarrow\infty$, one has that $\rho(x)$ does not depend on the starting point.
The cumulative distribution associated to $\rho$ results as
\begin{eqnarray}\label{2-9invariant3}
C(x)=\sum_{i=1}^{M}\left(\sum_{j=1}^{i}\rho_j\right) \ \chi_{[\frac{i-1}{M},\frac{i}{M})}(x) \nonumber
\end{eqnarray}

\subsection{Metric spaces}
We consider a set $\Gamma$ and a nonnegative function $d:\Gamma\times \Gamma\rightarrow\mathbb{R}_{+}$ satisfying the following axioms:
\begin{itemize}
  \item[$(a)$] Distinguishability: $d(x,y)=0$ \textrm{iff} $x=y$
  \item[$(b)$] Symmetry: $d(x,y)=d(y,x)$
  \item[$(c)$] Triangle inequality: $d(x,y)+d(y,z)\geq d(x,z)$
\end{itemize}
The function $d$ is called a \emph{metric} and the pair $(\Gamma,d)$ defines a \emph{metric space}. A standard example of metric space is $\mathbb{R}^m$ with the usual Euclidean metric.
Axiom $(a)$ refers to the intuitive idea of how to distinguish two points when they are at a nonzero distance from each other.
Axiom $(b)$ expresses that any legitimate distance must be symmetric with respect to the pair of points. Axiom $(c)$ is the well known Pythagoras triangle inequality and is crucial to extend many theorems and properties out of $\mathbb{R}^m$.
Relaxing $(b)$ and $(c)$ defines a \emph{divergence}, and relaxing only $(c)$ gives what is called a \emph{distance}.

An example is the well known Kullback--Leibler divergence in the context of information theory
as a relative entropy between two probability density functions \cite{kullback}. In quantum mechanics, the Jensen--Shannon divergence is an example of distance in Hilbert space and also can be used to define a measure of distinguishability and entanglement between quantum states \cite{plastino,lamberti1,lamberti2}. Moreover, using notions of entropy and purification one can also define metrics in quantum state spaces \cite{lamberti3}.
The concept of metric space arises as a generalization of the Euclidean space, many of its relevant properties and results concerning completeness, conexity, etc. have been extended for abstract spaces in general \cite{kolmogorov,kelley,lages}.
In this work we consider the state space $\Gamma$ as a subset of some Euclidean space $\mathbb{R}^m$.

\section{Metrics based on invariant densities}\label{sec_metric}
Based on the work by Johal \cite{johal} we recall the motivation to define a distance for discrete maps. From this we provide a metric in a mathematically strict sense, i.e. obeying axioms $(a)$--$(c)$. Then we define a `natural' metric space associated to a discrete map, induced by the metric.

\subsection{A metric for discrete maps}

Wootters proposed a distance between two probability distributions that gives a notion of distinguishability between states (classical or quantum) \cite{wooters}. A typical example is to consider two weighted coins represented by corresponding discrete probability distributions $p=(p_1,p_2)$ and $q=(q_1,q_2)$. The Wootters distance $D$ between $p$ and $q$ is defined as
\begin{eqnarray}\label{2-1}
&D(p,q)=\lim_{n\rightarrow\infty}\frac{1}{\sqrt{n}}\times [\textrm{maximum number of mutually} \nonumber\\
&\textrm{distinguishable intermediate probablities in} \ n \ \textrm{trials}] \nonumber
\end{eqnarray}
The relevant conclusion of Wootters' contribution has a double value: on the one hand this statistical distance can be defined on any probability space. On the other hand, and maybe more important, is the fact that $D(p,q)$ is involved in determining the geometry of the curved manifold of all distinguishable probability distributions.

For unidimensional discrete maps showing chaotic dynamics, inspired by Wootters ideas, in \cite{johal} was introduced a notion of distance $d$ between two points $x^A,x^B$ of state space $\Gamma$, given by
\begin{eqnarray}\label{2-2}
&d(x^A,x^B)=\lim_{n\rightarrow\infty} [\textrm{probability of visiting the interval} \nonumber\\
&\textrm{between} \ x^A \ \textrm{and} \ x^B \ \textrm{after} \ n \ \textrm{time steps} ] \nonumber
\end{eqnarray}
Now, if $\rho(x)$ is the invariant density of the discrete map under study, Johal argues that the distance $d(x^A,x^B)$ can be expressed as an integral of the invariant density on the interval. From Ref. \cite{johal}
\begin{eqnarray}\label{2-3}
&d(x^A,x^B)=\left|\int_{x^A}^{x^B}\rho(x)dx \right|
\end{eqnarray}

Using the properties of the integral, it is straightforward to check that $d(x^A,x^B)$ satisfies axioms $(b)$ and $(c)$. Concerning axiom $(a)$, there exist situations for which it is not fulfilled. For instance, if $\rho(x)=0$ in almost every point of some interval $[x^A,x^B]$ then from \eqref{2-3} it follows that $d(x^A,x^B)=0$ even when $x^A\neq x^B$. Strictly speaking, formula \eqref{2-3} does not define a metric unless one can guarantee that there is no interval $[x^A,x^B]$ for which $\rho(x)\equiv0$ a.e. And as we shall see, this condition cannot be guaranteed for all the regimes of the dynamics of a given map.

In order to solve this problem we propose to use equivalence classes instead of points in $\Gamma$. More precisely, we define the following relation $\sim$ in $\Gamma\times \Gamma$:
\begin{eqnarray}\label{2-4}
x \sim x^{\prime} \ \ \ \Longleftrightarrow \ \ \ d(x,x^{\prime})=0
\end{eqnarray}
Let us show that $\sim$ is an equivalence relation. By \eqref{2-3} one has $d(x^A,x^A)=0$ then $x^A\sim x^A$ for all $x^A\in \Gamma$. Due to symmetry of formula \eqref{2-3}, it follows that $x^A \sim x^B$ if and only if $x^B \sim x^A$ for all $x^A,x^B\in \Gamma$. If $x^A,x^B,x^C\in \Gamma$ with $x^A\sim x^B$ and $x^B\sim x^C$ then $d(x^A,x^B)=0=d(x^B,x^C)$, then by the nonnegativity of $d$ and the triangle inequality one has $0\leq d(x^A,x^C)\leq d(x^A,x^B)+d(x^B,x^C)=0$, so one has $d(x^A,x^C)=0$, i.e. $x^A \sim x^C$.

Considering the set of classes
$\Gamma/\!\sim \ = \{\overline{x}:x\in\Gamma\}$, where $\overline{x}=\{x^{\prime}\in\Gamma:x^{\prime}\sim x\}$, one can define rigorously a metric $\overline{d}$ between elements of $\Gamma/\sim$ in the following way. Let us consider two points $x^A,x^B$, then
\begin{eqnarray}\label{2-5}
&\overline{d}:\Gamma/\!\sim \ \times \ \Gamma/\!\sim \ \longrightarrow \ \mathbb{R}_{+} \nonumber\\
& \overline{d}(\overline{x^A},\overline{x^B}):=d(x^A,x^B)
\end{eqnarray}
being $\overline{x^A},\overline{x^B}$ the classes of $x^A$ and $x^B$ respectively. Let us show that $\overline{d}$ is well defined in the sense that it is independent of the representative elements chosen for each class. For arbitrary
$x\in \overline{x^A},x^{\prime}\in \overline{x^B}$ then one has $d(x,x^A)=d(x^{\prime},x^B)=0$. By applying the triangle inequality and the symmetry property
\begin{eqnarray}
&d(x,x^{\prime})\leq d(x,x^A)+d(x^A,x^B)+d(x^B,x^{\prime})\nonumber\\
&=d(x^A,x^B)\nonumber
\end{eqnarray}
and
\begin{eqnarray}
&d(x^A,x^B)\leq d(x^A,x)+d(x,x^{\prime})+d(x^{\prime},x^B)\nonumber\\
&=d(x,x^{\prime})\nonumber
\end{eqnarray}
Thus, $d(x,x^{\prime})=d(x^A,x^B)$ and by the definition \eqref{2-5} it follows that $\overline{d}(\overline{x},\overline{x^{\prime}})=\overline{d}(\overline{x^A},\overline{x^B})$.
Using \eqref{2-5} and since $d$ satisfies the axioms $(b)$ and $(c)$ then $\overline{d}$ also satisfies them. Now we can see that $\overline{d}$ satisfies the distinguishability axiom $(a)$. From definitions \eqref{2-4} and \eqref{2-5} one has
\begin{eqnarray}\label{2-6}
&\overline{d}(\overline{x}^A,\overline{x}^B)=0 \ \ \Longleftrightarrow \ \ d(x^A,x^B)=0 \nonumber\\
&\ \ \Longleftrightarrow \ \ x^A \sim x^B \ \ \Longleftrightarrow \ \ \overline{x}^A=\overline{x}^B \nonumber
\end{eqnarray}
Therefore, $\overline{d}$ is a metric on $\Gamma/\!\sim$.

\subsection{A metric space associated to the map dynamics}

Given a unidimensional discrete map characterized by $f:\Gamma\rightarrow\Gamma$, with $\Gamma\subseteq\mathbb{R}$, and the metric $\overline{d}$ induced by the invariant density $\rho:\Gamma\rightarrow\mathbb{R}_{+}$, we say that $(\Gamma/\sim,\overline{d})$ is the \emph{metric space associated to the map dynamics} generated by $f$.
In this way, given the pair $\{\Gamma,f\}$ one has the following `canonical' association between discrete maps and metric spaces:
\begin{eqnarray}\label{2-7}
\{\Gamma,f\} \ \longrightarrow \ \rho \ \longrightarrow \ d \ \longrightarrow \ (\Gamma/\!\sim,\overline{d})
\end{eqnarray}

In general, a discrete map can be given by a function $f_{\mathbf{r}}:\Gamma\rightarrow\Gamma$ depending on a set of real parameters $\mathbf{r}=~(r_1,\ldots,r_k)$ which play the role of controlling the type of dynamics generated. Thus, for a set of parameters $\mathbf{r}$ the map is given in the form:
\begin{eqnarray}\label{2-7bis}
x_{n+1}=f_{\mathbf{r}}(x_n) \ \ \ \ \forall \ n\in\mathbb{N}_0 \ \ , \ \ x_0\in\Gamma
\end{eqnarray}
Therefore, one has a parameterized family of correspondences of the form \eqref{2-7}, i.e.
\begin{eqnarray}\label{2-8}
\{\Gamma,f_{\mathbf{r}}\} \ \longrightarrow \ \rho_{\mathbf{r}} \ \longrightarrow \ d_{\mathbf{r}} \ \longrightarrow \ (\Gamma/\!\sim_{\mathbf{r}},\overline{d}_{\mathbf{r}})
\end{eqnarray}
where $d_{\mathbf{r}}$, $\overline{d}_{\mathbf{r}}$ and $\sim_{\mathbf{r}}$ are given as in Eqs. \eqref{2-3}, \eqref{2-4} and \eqref{2-5}, respectively.
Varying the parameter $\mathbf{r}$ one can study transitions in the dynamics of the map in terms of the corresponding changes in the topology of $(\Gamma/\!\sim_{\mathbf{r}},\overline{d}_{\mathbf{r}})$.

\subsection{Wandering set of a map in terms of the metric}
According to definitions \eqref{0-2}, \eqref{2-3} and by the properties of the integral it follows that
\begin{eqnarray}\label{2-8bis}
d(x^A,x^B)=\left|C(x^B)-C(x^A)\right|
\end{eqnarray}
where $C(x)$ is the cumulative distribution function of the invariant density $\rho$. This expression gives another characterization of the metric associated to a map that can be useful for studying the structure of the metric space. In fact, the so called \emph{wandering set} can be characterized in terms of the metric as follows. We recall the concept of \emph{wandering point} of a map.

Given a map generated by a continuous function $f:~\Gamma\rightarrow\Gamma$, a point $w\in\Gamma$ is said to be a \emph{wandering point} if there is a neighbourhood $U$ of $w$ and a positive integer $N$ such that for all $n>N$ one has
\begin{eqnarray}\label{2-8bisbis}
\mu(f^{n}(U)\cap U)=0 \nonumber
\end{eqnarray}
where $\mu$ is a measure defined over a $\sigma$--algebra $\Sigma$ of $\Gamma$ and $f^{n}(U)$ is the set $\{f^{n}(u):u\in U\}$. The \emph{wandering set}, which we denote as $\mathcal{W}(\Gamma)$, is defined as the set of all wandering points.

The following result relates the metric space associated to a map with its wandering set.
\begin{theorem}\label{theorem}
Let $\Gamma=[a,b]$ be the phase space of the map generated by $f:[a,b]\rightarrow[a,b]$ and let $[x^A,x^B]\subseteq~[a,b]$ be a subinterval of $[a,b]$. Then, considering the measure $\mu_{\rho}(A)=\int_{A}\rho(x)dx$ for all subset $A\subseteq[a,b]$ with $\rho(x)$ the invariant density of the map, the following propositions are equivalent:
\begin{itemize}
\item[$(a)$] $\rho(x)=0$ a.e. in $(x^A,x^B)$.
\item[$(b)$] $\overline{x}^A=\overline{x}^B$.
\item[$(c)$] $d(x^A,x^B)=0$.
\item[$(d)$] $C(x)$ is constant in $[x^A,x^B]$.
\item[$(e)$] All $x\in(x^A,x^B)$ is a wandering point.
\end{itemize}
\end{theorem}
\begin{proof}

$(a)\Longrightarrow(b):$ If $\rho(x)=0$ a.e. in $(x^A,x^B)$ then it is clear from \eqref{2-3} that $d(x^A,x^B)=0$, which means by definition that $\overline{x}^A=\overline{x}^B$.

$(b)\Longrightarrow(c):$ It follows by the definition \eqref{2-4}.

$(c)\Longrightarrow(d):$ Let us assume that $d(x^A,x^B)=0$. Then, by \eqref{2-8bis} one obtains $C(x^A)=C(x^B)$. Now, since $C(x)$ is an increasing function of $x$ then $C(x)$ is constant in $[x^A,x^B]$.

$(d)\Longrightarrow(e):$ Let $x$ be a point of $(x^A,x^B)$. Since $(x^A,x^B)$ is a neighbourhood of $x$ and $C(x)$ is constant in $[x^A,x^B]$ one has $C(x^B)-C(x^A)=\int_{x^A}^{x^B}\rho(x)dx=\mu_{\rho}((x^A,x^B))=0$. In particular, $f^{n}(x^A,x^B)\cap(x^A,x^B)\subseteq~ (x^A,x^B)$ for all $n\in\mathbb{N}$ from which follows that $\mu_{\rho}(f^{n}(x^A,x^B)\cap(x^A,x^B))=0$, i.e. $x$ is a wandering point.

$(e)\Longrightarrow(a):$ Let $x$ be a point of $(x^A,x^B)$. By hypothesis, there exist an neighbourhood $U$ of $x$ and a positive integer $N$ such that
$\mu_{\rho}(f^n(U)\cap U)=0$ for all $n>N$. In particular, one can repeat the histogram construction for $\rho$ by dividing the interval $[a,b]$ in $M$ subintervals $I_1,\ldots,I_M$ of equal length in a such way that $x,x_0\in I_{j_0}\subset U$ for some $j_0\in \{1,\ldots,M\}$. That is, one has the invariant density $\rho(x)=\sum_{i=1}^{M}\rho_i\chi_{I_i}(x)$ with $\chi_{I_i}(x)$ the characteristic function of $I_i$ for all $i=1,\ldots,M$.
Then, one has two situations that are mutually excluding.

If $f^{k}(x_0)\notin I_{j_0}$ for all $k>N$ then the subinterval $I_{j_0}$ is only finitely visited in the limit of large iterations $T\gg M$ which means that $\rho_{j_0}=0$, i.e. $\rho=0$ a.e. in $I_{j_0}$.

On the contrary, if $f^{k_0}(x_0)\in I_{j_0}$ for some $k_0>N$ then $f^{k_0}(I_{j_0})\cap I_{j_0}$ is a nonempty set, and therefore, since $f$ is assumed to be a continuous function it follows that $f^{k_0}(I_{j_0})\cap I_{j_0}$ is an nonempty interval. Thus, $|f^{k_0}(I_{j_0})\cap~ I_{j_0}|>0$ with $| \ . \ |$ the Lebesgue measure in $\mathbb{R}$. Moreover, since
$\mu_{\rho}(f^{k_0}(I_{j_0})\cap I_{j_0})=0$ one has
\begin{eqnarray}
&\int_{f^{k_0}(I_{j_0})\cap I_{j_0}}\rho(x)dx=\int_{f^{k_0}(I_{j_0})\cap I_{j_0}}\rho_{j_0}\chi_{I_{j_0}}(x)\nonumber\\
&=\rho_{j_0}|f^{k_0}(I_{j_0})\cap I_{j_0}|=0 \nonumber
\end{eqnarray}
with $|f^{k_0}(I_{j_0})\cap I_{j_0}|>0$. This implies that $\rho_{j_0}=0$ and then $\rho=0$ a.e. in $I_{j_0}$. In both cases one obtains $\rho=0$ a.e. in $I_{j_0}$ with with $x\in I_{j_0}\subset(x^A,x^B)$. Therefore, one has that $\rho=0$ a.e. in $(x^A,x^B)$.
\end{proof}
Physically, the wandering set is an important concept since when a dynamical system has a wandering set of nonzero measure, then the system is  \emph{dissipative}. Thus, the wandering set characterizes dissipation in dynamical systems and in maps. Theorem \ref{theorem} states that dissipation in discrete maps can be expressed by means of the metric or alternatively, by using the cumulative distribution function.

\section{Models and results}\label{sec_models}
In order to illustrate the ideas introduced in the previous section we consider two examples of chaotic maps: the logistic map and the circle map. The interest on the logistic map is, basically, that it contains all the features of the chaotic dynamics and the onset of chaos \cite{tirnakli}. The circle map describes a simplified model of the phase--locked loop in electronics, and also has been used to study the dynamical behavior of a beating heart \cite{circle1}, among other applications.

\subsection{The logistic map}
The logistic map is defined by the sequence \cite{may,strogatz}
\begin{eqnarray}\label{2-9}
x_{n+1}=rx_n(1-x_n) \ \ \ \ \forall \ n\in\mathbb{N}_0 \nonumber
\end{eqnarray}
i.e., in the form \eqref{2-7bis} where $f_r:[0,1]\rightarrow[0,1]$ is given by $f_r(x)=rx(1-x)$, depending on a unique parameter $r\in(0,4]$.
The dynamics has been characterized for all $r$. Here we focus on some special characteristic values which give rise to different relevant dynamics. Specifically, we consider $r$ equal to $1$, $2$, $3.56995\ldots$, $3.82843$ and $4$.

For $r=1$ one has a regular behavior where for any initial condition $x_0$, the sequence $\{x_n\}$ goes to zero, i.e. $x=0$ is an attractor for all point of $[0,1]$.

For $r=2$ the dynamics is also regular and all sequences $\{x_n\}$ tend to the fixed point $x=1/2$.

At $r=r_c=3.56995\ldots$ the onset of chaos occurs where from almost all initial conditions oscillations of finite period are not observed.

The value $r^*=3.82843$ corresponds to the region known as the ``Pomeau--Manneville scenario" characterized by a periodic (laminar) phase interrupted by bursts of aperiodic behavior, with a mixed dynamics composed by chaotic trajectories and \emph{islands of stability} that show non--chaotic behavior.

Finally, for $r=4$ the behavior is fully chaotic and the dynamics satisfies the mixing property, with no oscillations and the interval $(0,1)$ is an attractor for any initial condition $x_0$.


We calculate the invariant density for the logistic map.
For the regular cases $r=1$ and $r=2$, since all the sequences tend to only one fixed point, it follows that the invariant density is a delta function; this fact can also be argued by using the Frobenius--Perron operator and Eqs.~\eqref{1-2}.
Indeed, for $r=1$ and $r=2$ it is satisfied that $\rho_1(x)=\delta(x)$ and $\rho_2(x)=\delta(x-\frac{1}{2})$, i.e. $x=0$ and $x=\frac{1}{2}$ are attractors, respectively, for all initial condition; the only point that can be distinguished in each case is precisely the attractor. For $r=1$ all the interval $[0,1]$ collapses into two classes: $\overline{0}$ and $\overline{1}$, and for $r=2$ the classes are: $\overline{0}$, $\overline{\frac{1}{2}}$ and $\overline{1}$.
For the fully chaotic case $r=4$, as mentioned in Sec.~\ref{sec_prelim}, an analytic result is available; indeed the invariant density has the form $\rho_4(x)=\frac{1}{\pi\sqrt{x(1-x)}}$.

For the cases $r=r_c$ and $r^*$, and also for $r=4$, we computed numerically the invariant density following the method described in Section~2.
We carried out this procedure for an initial state $x_0=0.1$, and taking $M=10^3$ and $T=10^5$. The process was repeated for other initial states, leading the same result. This means that despite the sensitivity of trajectories to initial states, the invariant density remains the same when a large number of iterations is considered.
The results are given in Fig.~1.
In the fully chaotic case, one can see the agreement of the numerical computation with the analytical result for the invariant density~$\rho_4$.
\begin{figure}\label{fig:logistic}
  \centerline{\includegraphics[width=9.5cm]{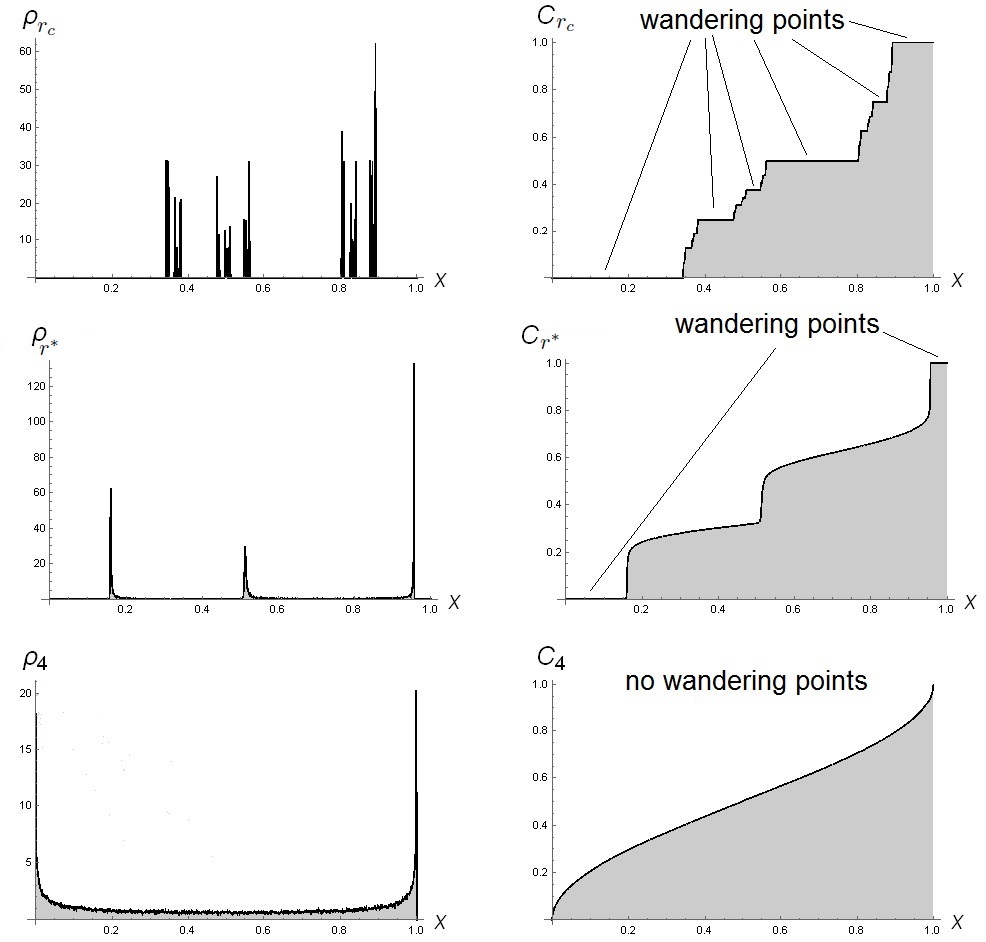}}
  \caption{Invariant densities (left column) and their cumulative distribution functions (right column) of the logistic map for $r=r_c$, $r=r^*$ and $r=4$ with $T=10^5$ and $M=10^3$. From Theorem~\ref{theorem} one has that the plateaus of $C_r(x)$ allow to identify the dissipative regions of the dynamics, i.e. the wandering points. The transition to a fully chaotic dynamics implies the total suppression of the wandering set as a consequence of the emergence of the chaotic sea.}
  \end{figure}

Now we focus on the study of  the metric space $([0,1]/\sim_{r},\overline{d}_r)$, the wandering set $\mathcal{W}_r([0,1])$, and the cumulative distribution function $C_r(x)$.
For $r=1, 2$, and 4 analytical results can be provided from Eqs.~\eqref{2-3}--\eqref{2-5} and Theorem~\ref{theorem}.

We obtain:
\begin{itemize}
  \item Case $r=1$: \\
   $[0,1]/\sim_{1} \ = \ \{\overline{0} \ , \ \overline{1}\}$ , with
  \begin{eqnarray}\label{2-10}
\overline{d}_{1}(\overline{0},\overline{1})=1 , \nonumber
\end{eqnarray}
\begin{eqnarray}\label{2-11}
C_1(x)=1 \quad \forall \ x\in[0,1] \nonumber
\end{eqnarray}
and
\begin{eqnarray}\label{2-12}
\mathcal{W}_1([0,1])=(0,1] . \nonumber
\end{eqnarray}
  \item Case $r=2$: \\
  $[0,1]/\sim_{2} \ = \ \{\overline{0} \ , \ \overline{\frac{1}{2}}\ , \ \overline{1}\}$ , with
\begin{eqnarray}\label{2-13}
\overline{d}_{2}(\overline{0},\overline{1})=\overline{d}_{2}(\overline{0},\overline{\frac{1}{2}})=\overline{d}_{2}(\overline{\frac{1}{2}},\overline{1})=1 , \nonumber
\end{eqnarray}
\begin{eqnarray}\label{2-14}
C_2(x)=2\,\Theta(x-1/2) \quad \forall \ x\in[0,1] \nonumber
\end{eqnarray}
and
\begin{eqnarray}\label{2-15}
\mathcal{W}_2([0,1])=[0,\frac{1}{2})\cup(\frac{1}{2},1] . \nonumber
\end{eqnarray}

\item Case $r=4$: \\
$[0,1]/\!\sim_{4} \ = \{\overline{x} : \ x\in [0,1]\}$ , with
\begin{eqnarray}\label{2-20}
\overline{d}_{4}(\overline{x}^A,\overline{x}^B)=|x^B-x^A| , \nonumber
\end{eqnarray}
\begin{eqnarray}\label{2-21}
C_4(x)=\frac{2}{\pi}\arccos(x) \quad \forall \ x\in[0,1] \nonumber
\end{eqnarray}
and
\begin{eqnarray}\label{2-21}
\mathcal{W}_4([0,1])=\emptyset . \nonumber
\end{eqnarray}

\end{itemize}

For $r=4$, we notice that making the change of variables $x=\textrm{sin}^2 (\frac{\pi y}{2})$, which is a diffeomorphism of $[0,1]$ on itself, one has \ $dx=\pi\,\textrm{cos}\frac{\pi y}{2}\textrm{sin}\frac{\pi y}{2}\,dy$. Therefore
\begin{eqnarray}\label{2-17}
&\rho_4(x)\,dx=\frac{1}{\pi\sqrt{\textrm{sin}^2 \frac{\pi y}{2}(1-\textrm{sin}^2 \frac{\pi y}{2})}} \, \pi\,\textrm{cos}\frac{\pi y}{2}\textrm{sin}\frac{\pi y}{2}\,dy = dy \nonumber
\end{eqnarray}
This expresses that the distance $d_4$ is equivalent to
\begin{eqnarray}\label{2-18}
d^{\prime}_4(y^A,y^B)=\left|\int_{y^A}^{y^B}dy\right|=|y^B-y^A| \nonumber
\end{eqnarray}
which is nothing but the Euclidean metric for the set of real numbers. Therefore, one can use $d^{\prime}_4$ instead of the metric that results from integrating the invariant density $\rho_4(x)$. The wandering set has been totally suppressed giving rise to the chaotic sea along the whole interval $(0,1)$, in agreement with Theorem \ref{theorem}.


Now, for the cases $r=r_c$ and $r=r^*$, we perform a qualitative analysis from the numerical results depicted in Fig.~1.
\begin{itemize}
  \item Case $r=r_c$:\\
  The complexity of dynamics is manifested by the substructures in the profile of the cumulative distribution $C_{r_c}(x)$, as shown in Fig. 1. From Theorem \ref{theorem} one can see that the plateaus of $C_{r_c}(x)$ account for the dissipative regions, i.e. the wandering points. As a consequence of the onset of chaos, the dynamics is mixed and composed by some stability islands represented by the increasing intervals of $C_{r_c}(x)$ immersed in the emergent chaotic sea. As soon as the parameter $r$ increases slightly the plateaus of $C_{r_c}(x)$ tend to disappear meaning that the transition to the fully chaotic regime is associated with a suppression of the wandering set.
 \item Case $r=r^*$: \\
 The only stability islands are given by the the few jump discontinuities observed in $C_{r^*}(x)$ while the wandering set has been reduced to $[0,0.15)\cup~ (0.95,1]$, as one can see from Fig. 1. The plateaus observed for $r_c$ now are part of the chaotic sea representing increasing intervals of the cumulative distribution $C_{r^*}(x)$.
\end{itemize}

\subsection{The circle map}

With regard the dynamics on the unitary circle $S^1$, Kolmogorov proposed a family of simplified models for driven mechanical rotors whose discretized equations define a map, called the \emph{circle map}. This is defined by the sequence
\begin{eqnarray}\label{2-21}
\theta_{n+1}=\theta_{n}+\Omega-\frac{K}{2\pi}\sin(2\pi\theta_n)\ \ \ \ \forall \ n\in\mathbb{N}_0 \nonumber
\end{eqnarray}
which has the form \eqref{2-7bis} where $f_{(\Omega,K)}:[0,1]\rightarrow[0,1]$ is given by $f_{(\Omega,K)}(\theta)=\theta+\Omega-\frac{K}{2\pi}\sin(2\pi \theta)$, depending on a two parameters, the coupling strength $K$ and the driving phase $\Omega$. Also, $S^1$ is represented by the interval $[0,1)$, i.e. the angle $\theta$ is expressed in units of $2\pi$.

When $K=0$ the map is a rotation in an angle $\Omega$. If one also chooses $\Omega$ irrational, the map reduces to an irrational rotation.
The dynamics has been studied for multiple combinations of the values $\Omega$ and $K$ exhibiting a very rich structure such as Cantor functions by plotting sections, subaharmonic routes of chaos, and the phenomena of phase locking. The complex structure of the dynamics has proved its usefulness in several applications \cite{circle1,circle2,circle3}.
Here we focus on some special characteristic values where the dynamics can be well distinguished. We consider the values $(\frac{1}{8},0)$, $(\frac{\sqrt{2}}{2\pi},0)$, and $(0,3)$ for the pair $(\Omega,K)$.

For $\Omega=\frac{1}{8}$ and $K=0$ one has a rotation in an angle $\frac{\pi}{4}$ with a regular behavior for any initial condition $\theta_0$. The sequence $\{\theta_n\}$ only takes the values $\theta_0$, $\theta_0+i\frac{1}{8}$ with $i=0,\ldots,7$, i.e. these ones are the only attractors for all initial condition.

For $\Omega=\frac{\sqrt{2}}{2\pi}$ and $K=0$ one has the irrational rotation which corresponds to an ergodic dynamics. In fact, for any initial condition $\theta_0$ the sequence $\{\theta_n\}$ travels around the entire phase space.

Finally, for $\Omega=0$ and $K=3$ the behavior is fully chaotic and the dynamics is mixing. As in the case of the irrational rotation, every sequence $\{\theta_n\}$ fill all the phase space but randomly.
\begin{figure}[!!ht]\label{fig2}
  \centerline{\includegraphics[width=8cm]{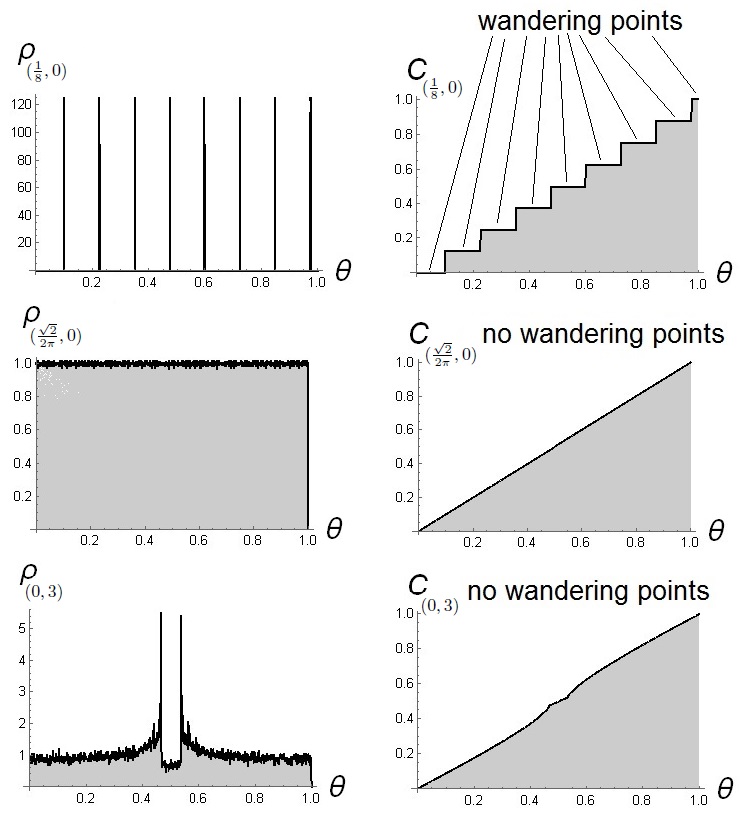}}
  \caption{Invariant densities (left column) and their cumulative distribution functions (right column) of the circle map for $(\Omega,K)=(\frac{1}{8},0)$, $(\Omega,K)=(\frac{\sqrt{2}}{2\pi},0)$ and $(\Omega,K)=(0,3)$ with $T=10^5$ and $M=10^3$. Again, the Theorem \ref{theorem} allows to identify the plateaus of $C_{(\Omega,K)}(\theta)$ with the dissipative regions of the dynamics. As in the logistic case, the transition from a regular dynamics to a chaotic one is associated to the supression of the wandering set.}
  \end{figure}

Carrying out the histogram method of Section \ref{sec_prelim} for an initial state $\theta_0=0.1$ and considering $M=10^3$ and $T=10^5$ we computed the invariant densities and their cumulative distribution functions for $(\Omega,K)=(\frac{1}{8},0)$, $(\Omega,K)=(\frac{\sqrt{2}}{2\pi},0)$ and $(\Omega,K)=(0,3)$. The results are shown in Fig. 2.
By applying \eqref{2-3}--\eqref{2-8} and the Theorem \ref{theorem} we obtain the metric space $(\Gamma/\sim_{(\Omega,K)}~,\overline{d}_{(\Omega,K)})$, the wandering set $\mathcal{W}_{(\Omega,K)}([0,1))$, and the cumulative distribution function $C_{(\Omega,K)}(\theta)$ for each one of the cases studied:
\begin{itemize}
\item Case $(\Omega,K)=(\frac{1}{8},0)$: \\
  $\Gamma/\sim_{(\frac{1}{8},0)}=$\\
  $\{\overline{\frac{i}{8}}:i=1,\ldots,8\}\cup \{\overline{\frac{2j+1}{16}}:j=0,\ldots,7\} $, with
  \begin{eqnarray}\label{2-22}
&\overline{d}_{(\frac{1}{8},0)}(\overline{\frac{i_1}{8}},\overline{\frac{i_2}{8}})=\frac{1}{8}(|i_1-i_2|+\chi_{(0,+\infty)}(|i_1-i_2|))  \nonumber\\
& \forall \ \ i_1,i_2=1,\ldots,8 \nonumber\\
&\nonumber\\
&\overline{d}_{(\frac{1}{8},0)}(\overline{\frac{2j_1+1}{16}},\overline{\frac{2j_2+1}{16}})=\frac{1}{8}|j_1-j_2| \nonumber\\
&\ \forall \ \ j_1,j_2=0,\ldots,7 \nonumber\\
\nonumber\\
&\overline{d}_{(\frac{1}{8},0)}(\overline{\frac{i_1}{8}},\overline{\frac{2j_2+1}{16}})=\frac{1}{8}(|i_1-j_2|+\delta_{|i_1-j_2|,0}) \nonumber\\
&\forall \ \ i_1=1,\ldots,8 , \ \ j_2=0,\ldots,7 \nonumber
\end{eqnarray}
\begin{eqnarray}\label{2-23}
C_{(\frac{1}{8},0)}(\theta)=\frac{i}{8} \ \ \ \forall \ \theta\in[\frac{i}{8},\frac{i+1}{8}) \ \ , \ \ i=0,\ldots,7\nonumber
\end{eqnarray}
and
\begin{eqnarray}\label{2-24}
\mathcal{W}_{(\frac{1}{8},0)}([0,1))=\{\overline{\theta} : \ \theta\in [0,1)\}-\left\{\overline{\frac{i}{8}}:i=1,\ldots,8\right\} \nonumber
\end{eqnarray}
\item Cases $(\Omega,K)=(\frac{\sqrt{2}}{2\pi},0)$ and $(\Omega,K)=(0,3)$: \\
$\Gamma/\sim_{(\frac{\sqrt{2}}{2\pi},0)}=\Gamma/\sim_{(0,3)}=\{\overline{\theta} : \ \theta\in [0,1)\}$ , with
\begin{eqnarray}\label{2-25}
\overline{d}_{(\frac{\sqrt{2}}{2\pi},0)}(\overline{\theta}_1,\overline{\theta}_2)=\overline{d}_{(0,3)}(\overline{\theta}_1,\overline{\theta}_2)=|\theta_1-\theta_2| \nonumber
\end{eqnarray}
\begin{eqnarray}\label{2-26}
C_{(\frac{\sqrt{2}}{2\pi},0)}(\theta)=C_{(0,3)}(\theta)=\theta \ \ \ , \ \ \ \forall \ \theta\in[0,1) \nonumber
\end{eqnarray}
and
\begin{eqnarray}\label{2-27}
\mathcal{W}_{(\frac{\sqrt{2}}{2\pi},0)}([0,1))=\mathcal{W}_{(0,3)}([0,1))=\emptyset  \nonumber
\end{eqnarray}
\end{itemize}
As in the logistic case, one can see that the transition from a regular dynamics (the rotation in an angle $\frac{\pi}{4}$ for $\Omega=\frac{1}{8}$ and $K=0$) to a chaotic one (the ergodic irrational rotation for $\Omega=\frac{\sqrt{2}}{2\pi},K=0$ and the mixing regime $\Omega=0,K=3$) is associated to a suppression of the wandering set.

Moreover, from Fig. 2 one finds that the two levels of chaos involved in the cases, ergodic for $\Omega=\frac{\sqrt{2}}{2\pi},K=0$ and mixing for $\Omega=0,K=3$, can be only distinguished by means of the profile of their respective invariant densities since the cumulative distributions, the wandering sets, and the metric spaces are the same in both cases.
An explanation of this fact based on the statistical information of the dynamics of the map contained in its invariant density can be given as follows.
The signatures of chaos in the invariant density are expressed as statistical fluctuations observed in its profile which are more pronounced in the mixing case than in the ergodic case, as one can see from Fig. 2. By integrating the invariant density one averages these fluctuations and then they disappear. Thus, one obtains the same cumulative distribution in both cases, and therefore, the wandering set and the metric space also must be remain the same.

\vspace{0.2cm}
Our analytical results about the characterization of the map dynamics of the logistic and the circle maps are summarized in the Table I.

\begin{table}[!hbt]
\begin{center}
\begin{tabular}{|c | c | c|}
\hline
parameters & metric space $\Gamma/\!\sim_{\vec{r}}$ & wandering set $\mathcal{W}(\Gamma)$  \\
\hline
r=1 & $\{\overline{0} \ , \ \overline{1}\}$ & (0,1] \\
\hline
r=2&$\{\overline{0},\overline{\frac{1}{2}},\overline{1}\}$  & $[0,\frac{1}{2})\cup(\frac{1}{2},0]$\\
\hline
r=4&$\{\overline{x} : \ x\in [0,1]\}$& $\emptyset$\\
\hline
$\Omega=\frac{1}{8},K=0$&$\{\overline{\frac{i}{8}}:i=1,\ldots,8\}\cup$& $\{\overline{\theta} : \ \theta\in [0,1)\}$\\
&$\{\overline{\frac{2j+1}{16}}:j=0,\ldots,7\}$  &$-\left\{\overline{\frac{i}{8}}:i=1,\ldots,8\right\}$  \\
\hline
$\Omega=\frac{\sqrt{2}}{2\pi},K=0$&$\{\overline{\theta} : \ \theta\in [0,1)\}$& $\emptyset$\\
\hline
$\Omega=0,K=3$&$\{\overline{\theta} : \ \theta\in [0,1)\}$& $\emptyset$\\
\hline
\end{tabular}
\caption{Relationship between the metric space and the wandering set for the logistic and the circle maps. In both cases it is observed that the dissipation (measured by the wandering set) and the chaotic behavior are contrary properties, i.e. when increasing one the other decreases and viceversa, as a consequence of Theorem \ref{theorem}.}
\end{center}
\end{table}

\section{Conclusions}\label{sec_conclusions}

We proposed a notion of distinguishability for points of discrete maps of arbitrary dimension
in the limit where the number of iterations tends to infinity. Basing us in a statistical distance for unidimensional maps defined by means of the invariant density and using equivalence classes instead of points, we redefined it as a metric satisfying the axioms of distinguishability, symmetry and the triangular inequality. Furthermore, from this metric we defined a metric space associated to the map
which is composed by all the points that can be distinguished among the map dynamics in the limit of large number of iterations.
The importance of the metric space is that by means of changes in its structure it allows one to study the transitions of the dynamics, as the map parameters vary.

Complementary, we proved the Theorem \ref{theorem} that characterizes the wandering set of a map, which measures the dissipation present in the dynamics, in terms of the metric space and the cumulative distribution function of the invariant density. As consequence of the Theorem \ref{theorem}, we obtained that the size of the associated metric space increases as the corresponding to the wandering set decreases and viceversa.

By analyzing the dynamics of logistic map and of the circle map in terms of the associated metric space and of the wandering set for some characteristic values of their parameters, we illustrated the relevance of the associated metric space and the content of Theorem \ref{theorem} which lead us to the following facts:
\begin{itemize}
\item In both maps, when the metric space is finite the dynamics is regular and the wandering set is the whole phase space except for some discrete points that correspond to the classes composed by a single element. These single classes are precisely the only attractors of the dynamics.
\item In both maps, the chaotic regimes studied are characterized by a maximal size of the metric space and a minimal size of the wandering set. In their fully chaotic cases, i.e. $r=4$ in the logistic map and $\Omega=0,K=3$ in the circle map, the metric space is all the entire phase space and the wandering set is empty, corresponding to the absent of dissipation in the dynamics.
\item In the case of the logistic map the complex structure in the onset of chaos when $r=r_c$ is expressed by substructures in the profile cumulative distribution function of the invariant density.
\item In both maps it is observed that the transition from a less chaotic dynamics to a more chaotic one is associated to a suppression of the wandering set, making plausible to consider this fact as a signature of chaos of the dynamics of a map.
\item In the case of the circle map, the impossibility of distinguish an ergodic dynamics ($\Omega=\frac{\sqrt{2}}{2\pi},K=0$) from a mixing one ($\Omega=0,K=3$) is explained on the basis that by integrating statistical fluctuations of different invariant densities can lead to the same cumulative distribution function, and consequently to the same metric space and wandering set associated with the map dynamics.
\end{itemize}
\begin{figure}[!!ht]\label{fig3}
  \centerline{\includegraphics[width=9cm]{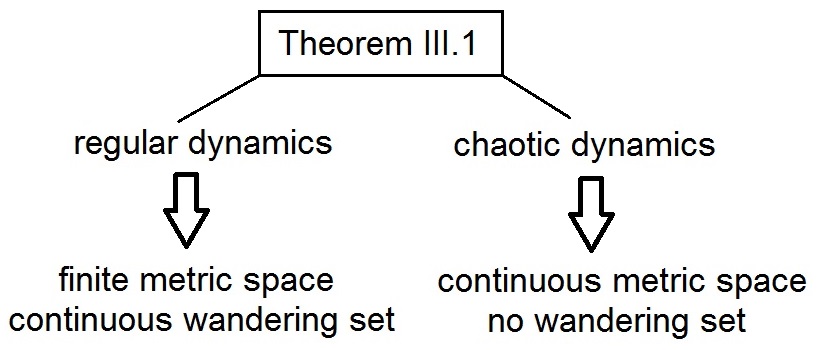}}
  \caption{An schematic picture of Theorem III.1 illustrating its consequences with regard the dynamics of the maps studied, i.e. the logistic and the circle maps.}
  \end{figure}

From all these facts we can consider that the proposed metric space associated to the dynamics of a map is a good and novel indicator along with Theorem \ref{theorem}, for characterizing and distinguishing a regular dynamics from a chaotic one. We hope that this conclusion can be analyzed with more examples, and for more dimensions than one, in future researches.

\section*{ACKNOWLEDGMENTS}
This work was partially supported by CONICET, Universidad Nacional de La Plata and Universidad Nacional de C\'{o}rdoba, Argentina.

\end{document}